\documentclass{iopjournal}
\usepackage{amssymb,amsmath}
\usepackage{ragged2e}
\usepackage{orcidlink}
\renewcommand{\orcid}[1]{\orcidlink{#1}}

\AtBeginDocument{\justifying}
\renewcommand{\headrulewidth}{0pt}
\renewcommand{\footrulewidth}{0pt}
\renewcommand{\articletype}[1]{}

\newtheorem{definition}{Definition}

\newtheorem{theorem}{Theorem}

\newcommand{\VV}{V-V}

\begin{document}

\articletype{Paper}

\title{Diverse Efficiency of Observable Optimization for Four-Level Quantum Systems with Higher-Order Traps}

\author{Alexander~N. Pechen$^{1,*}$\orcid{0000-0001-8290-8300} and 
Boris~O. Volkov$^{1}$\orcid{0000-0002-6430-9125}}

\affil{$^1$Department of Mathematical Methods for Quantum Technologies, Steklov Mathematical Institute of Russian Academy of Sciences, Gubkina str. 8, Moscow 119333, Russia}

\affil{$^*$Author to whom any correspondence should be addressed.}

\email{apechen@gmail.com}


\vspace{3mm}
\noindent{\fontsize{8}{10}\selectfont\textit{Author's Original / Preprint. Published version:}
A.~N. Pechen and B.~O. Volkov, \textit{Physica Scripta} \textbf{101} (2026) 025101.
DOI: \href{https://doi.org/10.1088/1402-4896/ae2f3d}{10.1088/1402-4896/ae2f3d}.}

\begin{abstract}
\justifying
In this work, we perform an analytical and numerical analysis of quantum landscapes for controlling  special four-level quantum systems for which we prove that the null control is a five-order trap: a $V-V$ system and an anharmonic system. As a control goal, an observable optimization is considered. The rigorous theoretical analysis is followed by the numerical experiments based on the GRadient Ascent Pulse Engineering (GRAPE) algorithm and Gradient Projection Method (GPM), performed to investigate the behavior of the efficiency of optimization for unconstrained (using GRAPE) and constrained (using GPM) controls. As the main result, we observe an interesting phenomenon with a diverse behavior of the optimization efficiency depending on the system Hamiltonian --- sharp increase of the optimization efficiency up to $100\%$ at certain distance from the null control for a $\VV$ system, while much slower and less significant increase (and even small decrease) for a system with the chain interaction.  This sharp difference might be related with the fine structure of the subspace of controls where second derivative of the objective functional is zero.
\end{abstract}

\section{Introduction}

Quantum control is an important area of research with fundamental relevance to determination of the ultimate degree of manipulating quantum systems and with various existing and prospective applications in quantum technologies, including quantum computation, laser assisted chemistry, design of NMR pulse sequences, etc.~\cite{KochEPJ2022,TannorBook2007,UG_BrifNewJPhys2010,Gough2012}. Various optimization methods are used to find optimal controls in the laboratory or numerical experiments, including GRadient Ascent Pulse Engineering (GRAPE) optimization algorithm~\cite{khaneja_optimal_2005}, gradient flows~\cite{Glaser2010}, Krotov method~\cite{Tannor1992}, genetic algorithms for coherent control of closed systems~\cite{Judson1992} and incoherent control of open quantum systems~\cite{PechenRabitzPRA2006}, gradient-free CRAB optimization algorithm ~\cite{CanevaPRA2011}, Hessian-based optimization as in the Broyden--Fletcher--Goldfarb--Shanno (BFGS) algorithm and combined approaches ~\cite{EitanPRA2011,DalgaardPRA2020}, etc. The efficiency of such methods strongly depends on the properties of the underlying quantum control landscape. Quantum control landscapes have multiple applications in physics and chemistry. They were exploited to manipulate the intensity of the Autler-Townes components in the photoelectron spectrum~\cite{Wollenhaupt_Prakelt_Sarpe-Tudoran_Liese_Baumert_2005}, making an experimental implementation for retinal photoisomerization in bacteriorhodopsin~\cite{Marquetand_Nuernberger_Vogt_Brixner_Engel_2007}, manipulating molecular systems~\cite{Ruetzel_Stolzenberger_Fechner_Dimler_Brixner_Tannor_2010}, inducing multi-photon excitations in atoms and vibrational population transfer in molecules~\cite{Palao_Reich_Koch_2013}, discovering the failure of greedy algorithms to generate fast quantum gates~\cite{Zahedinejad_Schirmer_Sanders_2014}, experimentally observing saddle points~\cite{Sun_Pelczer_Riviello_Wu_Rabitz_2015} and analyzing quantum state preparation and entanglement creation~\cite{Li_2023} in two spin quantum systems, dressing chopped-random-basis optimization~\cite{Rach_Muller_Calarco_Montangero_2015}, discovering a discontinuous phase transition with broken symmetry in a two-qubit quantum system~\cite{Bukov_Day_Weinberg_Polkovnikov_Mehta_Sels_2018}, etc.

In this work, we consider control landscapes for four-level quantum systems which evolve under the action of a coherent control $f(t)$ according to the Schr\"odinger equation for the unitary evolution operator $U_t^f$,
\begin{equation}
\label{Shroedinger}
i\frac{dU_t^f}{dt}=(H_0+f(t)V)U_t^f,\, \quad U_{t=0}^f = \mathbb{I}.
\end{equation}
Here $H_0$ and $V$ are the free and interaction Hamiltonians ($4\times 4$ Hermitian matrices). For such systems, we consider quantum control problem of maximizing mean value of some Hermitian observable $O=O^\dagger$ at some final time $T>0$,
\begin{equation}
\label{Objective}
J_O(f)={\rm Tr}[U_T^f\rho_0 U_T^{f\dagger} O]\to\max.
\end{equation}

Quantum control landscape of this problem is the graph of the objective functional $J_O(f)$. Its most important points are global maxima, which represent the desired controls, and traps, which represent controls which locally look like optimal but which are not globally optimal. Properties of the quantum control landscape are important for practical optimization, since presence of traps could impede for the search for the globally optimal controls~\cite{RHR}. For this reason, the analysis of quantum control landscapes attracts high attention. Absence of traps was proved for two-level quantum systems in~\cite{Pechen_Ilin_2014}. Trapping behavior was found for various multi-level quantum systems in~\cite{PechenTannor2011,PechenTannorReply,deFouquieresSchirmer,VolkovPechenUMN}. Properties of quantum control landscapes were investigated in various works, e.g. in~\cite{PechenRabitzEPL2010,PechenBrifPRA2010,Nanduri_Donovan_Ho_Rabitz_2013,Wu_Long_Dominy_Ho_Rabitz_2012,Zahedinejad_Schirmer_Sanders_2014,Rach_Muller_Calarco_Montangero_2015,Kosut_Arenz_Rabitz_2019}. Quantum control landscape for a two-level system near the quantum speed limit was described in~\cite{Larocca2018}. While for simplest quantum systems such as general two-level~\cite{Pechen_Ilin_2014} or special four-level~\cite{deFouquieresSchirmer} quantum systems some analytical results can be obtained, in many situations a numerical analysis of the quantum control landscape has to be performed~\cite{VolkovJPA2021,Dalgaard_Motzoi_Sherson_2022,Volkov_Myachkova_Pechen_2025,Fentaw_Campbell_Caton_2025}. To avoid local extrema, global search methods such as genetic algorithms~\cite{JudsonRabitzPRL1992,PechenRabitzPRA2006}, evolutionary algorithms~\cite{Zahedinejad_Schirmer_Sanders_2014} and other methods could be exploited. Recently, phenomenon of a stronger trapping behavior in three-level quantum systems with symmetry (so called $\Xi$-systems) was discovered analytically and investigated numerically~\cite{Volkov_Myachkova_Pechen_2025}. The order of trap, which is determined by the Taylor expansion of the objective functional \cite{PechenTannor2012,VolkovPechenUMN}, was shown to significantly affect the level of difficulty of practical optimization. 

In the present work, we analytically show that the null control is a five-order trap in quantum control landscapes for special four-level quantum systems: a $V-V$ system and an anharmonic system. Energy level structure of both quantum systems is shown on Fig.~\ref{Fig}. The system on the left we call as $\VV$ since it has a couple of $V$-connected triples of levels: One triple is level $|1\rangle$ connected with two levels $|2\rangle$ and $|3\rangle$, an another triple is level $|4\rangle$ connected with the same two levels; $''-''$ denotes that both upper levels $|2\rangle$ and $|3\rangle$ are concatenated.  The system on the right is an anharmonic system. Beyond theoretical analysis of the properties of quantum control landscapes for general classes of such systems, we perform a numerical analysis for particular examples considering both cases of unconstrained optimization using GRAPE and constrained optimization using GPM. The analysis of the control landscape for unconstrained optimization is based on the GRadient Ascent Pulse Engineering (GRAPE) method which was developed for coherent control of closed and open quantum systems in~\cite{Khaneja_Reiss_Kehlet_Schulte-Herbruggen_Glaser_2005} and for general case of coherent and incoherent control of open quantum systems in~\cite{PetruhanovPechenJPA2023}. The analysis of the control landscape for constrained optimization is based on the fixed step Gradient Projection Method (GPM) which was developed for coherent and incoherent control of open quantum systems using parametrization of Kraus maps by points of Stiefel manifolds~\cite{Oza2005} and for dynamic quantum control in~\cite{MorzhinPechenJPA2025}, and on the method of the analysis of constrained quantum control landscapes  using GPM developed in~\cite{PechenLJM2025}. As the main result, we observe a very interesting phenomenon  when optimization efficiency behaves differently depending on the system Hamiltonian. For a $\VV$ system we observe sharp increase of the optimization efficiency up to $100\%$ at certain distance from the null control. For the system with the chain interaction the observed increase in efficiency if much slower and less significant. This sharp difference might be related with the fine structure of the subspace of controls where the second derivative of the objective functional is zero.

The structure of this work is the following. In Sec.~\ref{Sec:2}, some classes of four-level quantum systems with null control which is  a five-order trap are considered and the theoretical analysis is performed culminated in a theorem about existence of five-order traps in such systems. In Sec.~\ref{Sec:3}, a numerical analysis using GRAPE of the unconstrained quantum control landscapes in a vicinity of the five-order trap for two examples of such classes of systems is performed, which shows the diverse behavior of the optimization efficiency for the two kinds of systems. In Sec.~\ref{Sec:GPM}, a numerical analysis using GPM of the respective constrained quantum control landscapes for the same systems is performed. Conclusions Sec.~\ref{Sec:4} summarizes the results.

\section{Existence of five-order traps in quantum control landscapes for the four-level systems}\label{Sec:2}

Consider objective functional~(\ref{Objective}) with target observable of the general form $O=\mathrm{diag}\{\lambda_1,\lambda_2,\lambda_3,\lambda_4\}$ with $\lambda_1>\lambda_4$,
$\lambda_4>\lambda_2$ and $\lambda_4>\lambda_3$ (in the simulations we take $O=\mathrm{diag}\{1,-\lambda,-\lambda,0\}$ with $\lambda=6$). As the initial state of the system consider the pure state $
\rho_0=|4\rangle \langle 4|$.
Following~\cite{PechenTannor2012,VolkovPechenUMN,Volkov_Myachkova_Pechen_2025}, consider the free Hamiltonian $H_0=\mathrm{diag}(h_1,h_2,h_3,h_4)$, and the interaction Hamiltonian $V$ of the form
\begin{equation}
\label{V}
V=\left(\begin{array}{cccc}
0& v_{12} & v_{13} & 0 \\
v^\ast_{12} &  0 & v_{23}  & v_{24} \\
v^\ast_{13}& v^\ast_{23}& 0 &  v_{34} \\
0& v^\ast_{24}& v_{34}^\ast & 0
\end{array}
\right).
\end{equation}
An example class of such Hamiltonians is shown on the left part of Fig.~\ref{Fig}. On the right  part of Fig.~\ref{Fig}, we also show an example of a four-level quantum system with the chain Hamiltonian which, as was established in~\cite{VolkovPechenUMN}, has trap of fifth order. 

\begin{figure}
    \includegraphics[width=\linewidth]{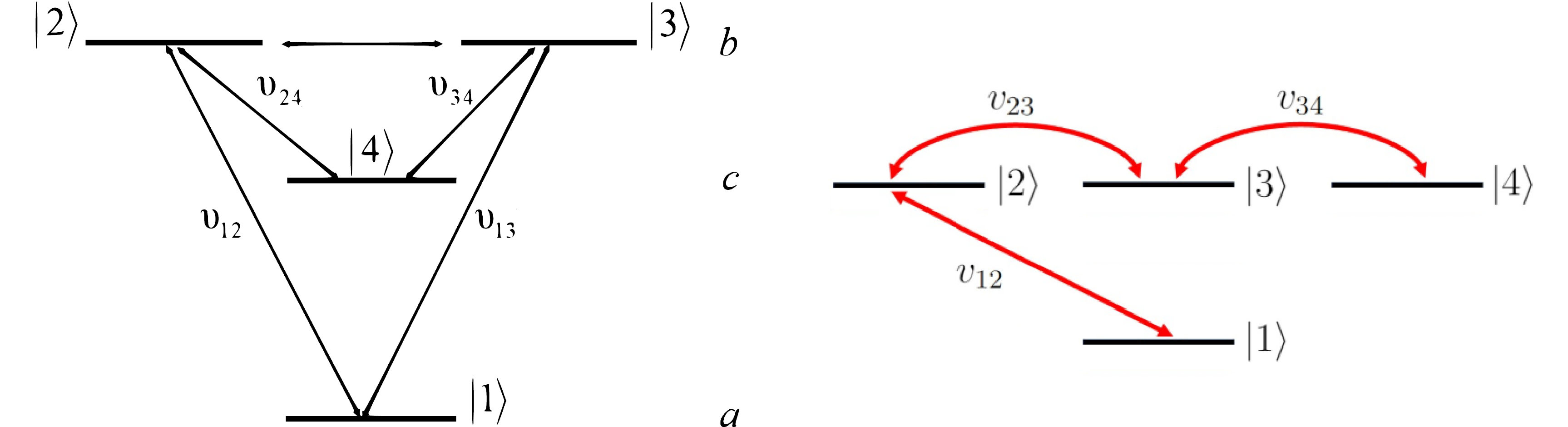}
    \caption{Left: Energy level structure of an example class of four-level quantum systems with the free Hamiltonian $H_0={\rm diag}(a,b,b,c)$ and with allowed transitions $|1\rangle\leftrightarrow|2\rangle$, $|1\rangle\leftrightarrow|3\rangle$, $|2\rangle\leftrightarrow|3\rangle$,  $|2\rangle\leftrightarrow|4\rangle$, and $|3\rangle\leftrightarrow|4\rangle$. We denote this system as a $\VV$ since it has a couple of $V$-connected triples of levels: One triple is level $|1\rangle$ connected with two levels $|2\rangle$ and $|3\rangle$, an another triple is level $|4\rangle$ connected with the same two levels; $''-''$ indicates that both upper levels $|2\rangle$ and $|3\rangle$ are connected.    
    Right: Energy level structure of a four-level quantum system with the chain interaction.}
    \label{Fig}
\end{figure}

For the analytical study we consider the control space $\mathfrak{H}^{0}\colon=L_2([0,T],\mathbb{R})$, so that coherent control in~(\ref{Shroedinger}) is  $f\in \mathfrak{H}^{0}$.
We also assume that the quantum system $(H_0,V)$ is controllable, i.e. there exists a time $T_{0}>0$ such that for any time $T > T_{0}$ and for any $U \in U(4)$ there exists a control $f \in L_2([0,T],\mathbb{R})$ such that $U = e^{i \alpha} U_T^f$ for some $\alpha \in [0, 2\pi)$. For the selected above~$\rho_0$ and~$O$, the complete controllability of the quantum system  $(H_0,V)$  implies that the control $f_0\equiv0$ is not a point of global extrema of~$J_O$ for any $T>T_{0}$. It was shown in~\cite{PechenTannor2012} that for such $(H_0,V)$,~$\rho_0$ and~$O$  the null control $f_0\equiv0$ is a trap of at least the second order, i.e. it is the critical point of the objective functional with  a positive semi-definite Hessian. Below we investigate some sufficient conditions for this trap to have order at least 5.

\begin{definition}
A control $f \in \mathfrak{H}^{0}$ is a trap of $n$-th order for the objective functional $J_O$, if $f$ is not a  global maximum point of $J_O$ and the Taylor expansion of the objective functional at the point $f$ has the form 
\[
J_O(f + \delta f) = J_O(f) + \sum\limits_{j=2}^n \frac{1}{j!}J_O^{(j)}(f)(\delta f, \dots, \delta f) + o(\|\delta f\|^n), \quad \|\delta f\| \rightarrow 0,
\]
where the polynomial $R(\delta f) = \sum\limits_{j=2}^{n}\frac{1}{j!} J_O^{(j)}(f)(\delta f, \dots, \delta f)$ satisfies the conditions:
\begin{enumerate}
\item  $R\neq0$;
\item for any $\delta f \in \mathfrak{H}^{0}$ there exists $\epsilon>0$ such that $R(t\delta f) \leq 0$ for all $t \in (-\epsilon, \epsilon)$.
\end{enumerate}
\end{definition}

Let $V_t:=e^{itH_0}Ve^{-itH_0}$. We can express the evolution operator $U_T^{f}$ for any $f\in \mathfrak{H}^{0}$ via the Dyson series as
\begin{equation}
\label{DysonSeries}
    U_T^{f}= e^{-iTH_0} \left(\mathbb{I} + \sum\limits_{n=1}^{\infty} (-i)^n \int\limits_0^T dt_1 \dots \int\limits_0^{t_{n-1}}dt_n  f(t_1) \dots f(t_n) V_{t_1} \dots V_{t_n}\right).
\end{equation}
By substituting expression~(\ref{DysonSeries}) into the objective functional $J_O(f)={\rm Tr} [U_T^f\rho_0 U_T^{f\dagger} O]$, we obtain the Taylor expansion of this objective functional at the null control $f_0$,
$J_O(f_0+\delta f)=\sum\limits_{n=0}^\infty \frac 1{n!}\ J^{(n)}_O(f_0)(\delta f,\ldots,\delta f)$. Consider the form $A^n_{lk}\colon \mathfrak{H}^0 \to\mathbb{C}$ of the order $n\in\mathbb  N$ defined as
\[
A^n_{lk}\langle f\rangle:=\displaystyle\int_0^Tdt_1\displaystyle\int_0^{t_1}dt_2\ldots
\displaystyle\int_0^{t_{n-1}}dt_n\,f(t_1)\cdots f(t_n)\langle l|V_{t_1}\cdots
V_{t_n}|k\rangle.
\]
For $n=0$ let $A^0_{lk}=\delta_{lk}$, where $\delta_{lk}$~is the Kronecker symbol. Direct calculations shows that the Fr\'echet differential of $(2n-1)$-th order of the objective functional $J_O$ at $f_0$ has the form
\begin{eqnarray}
\label{J(2n-1)}
\frac 1{(2n-1)!}J^{(2n-1)}_O(f_0)(\delta f,\ldots,\delta f)\nonumber\\=\sum_{l=1}^3\sum_{m=0}^{n-1}(\lambda_l-\lambda_4)(-1)^{n+m}i\left(\overline{A^{m}_{l4}\langle \delta f \rangle}A^{2n-1-m}_{l4}\langle \delta f \rangle-A^{m}_{l4}\langle \delta f \rangle \overline{A^{2n-1-m}_{l4}\langle \delta f \rangle}\right).
\end{eqnarray}
The Fr\'echet differential of $2n$-order of the objective functional $J_O$ at $f_0$ has the form
\begin{eqnarray}
\label{J(2n)}
\frac 1{(2n)!}J^{(2n)}_O(f_0)(\delta f,\ldots,\delta f)
=\sum_{l=1}^3(\lambda_l-\lambda_4)|A^{n}_{l4}\langle \delta f\rangle|^2\\+\sum_{l=1}^3\sum_{m=0}^{n-1}(\lambda_l-\lambda_4) (-1)^{n+m}(A^{m}_{l4}\langle \delta f \rangle \overline{A^{2n-m}_{l4}\langle \delta f \rangle}+\overline{A^{m}_{l4}\langle \delta f \rangle}A^{2n-m}_{l4}\langle \delta f \rangle).
\end{eqnarray}

For $n=1$, formula~(\ref{J(2n-1)}) implies that $J'_O(f_0)=0$. Since $A^1_{14}=0$, formula~(\ref{J(2n)}) implies that the second derivative of the objective functional is
\[
\frac 1{2!}J_O^{(2)}(f_0)(\delta f,\delta f)=(\lambda_3-\lambda_4)|A^1_{34}\langle \delta f\rangle|^2+(\lambda_2-\lambda_4)|A^1_{24}\langle \delta f\rangle|^2\leq 0.
\]
Thus, $f_0$ is indeed a trap of at least of the second order, as it was shown in~\cite{PechenTannor2011}. Consider the linear space of controls
\[
\mathfrak{H}^{1}
:=\{f \in \mathfrak{H}^{0}\colon  A^{1}_{34}\langle f \rangle=0, A^{1}_{24}\langle f \rangle=0 \}=\{f\colon J_O^{(2)}(f_0)(f,f)=0\}.
\] 
Since $A^1_{14}=0$, for $n=2$ formula~(\ref{J(2n-1)})  implies that the third order Fr\'echet differential vanishes on $\mathfrak{H}^{1}$
and for all $\delta f \in\mathfrak{H}^{1}$ holds that
\[
\frac{1}{4!} J_O^{(4)}(f_0)(\delta f,\ldots,\delta f)=(\lambda_1-\lambda_4)|A^2_{14}\langle \delta f\rangle|^2+(\lambda_3-\lambda_4)|A^2_{34}\langle \delta f\rangle|^2+(\lambda_2-\lambda_4)|A^2_{24}\langle \delta f\rangle|^2.
\]

The following theorem provides sufficient conditions for a trap to be at least of the fifth order.
\begin{theorem}
Let $H_0=\mathrm{diag}(a,b,b,c)$, i.e. $h_1=a$, $h_2=h_3=b$, $h_4=c$ (see~Fig.~\ref{Fig}), and $V$ having the form~(\ref{V}) with
\begin{equation}
\label{condvv}
v_{13}v_{34} +v_{12}v_{24}=0.
\end{equation}
Then the null control $f_0\equiv 0$ is a trap of at least of fifth order.
\end{theorem}
\noindent{\bf Proof.}
Under the assumptions of the theorem one has
\[
A^2_{14}\langle f\rangle=(v_{13}v_{34} +v_{12}v_{24})\int_0^1dt_1\int_0^{t_1}dt_2e^{-i(a-b)t_1-i(b-c)t_2}f(t_1)f(t_2)=0.
\]
Hence, for any $\delta f \in\mathfrak{H}^{1}$ it holds that
\[
\frac 1{4!}J_O^{(4)}(f_0)(\delta f,\ldots,\delta f)=(\lambda_3-\lambda_4)|A^2_{34}\langle \delta f\rangle|^2+(\lambda_2-\lambda_4)|A^2_{24}\langle \delta f\rangle|^2\leq 0.
\]
Note that according to the assumptions of the theorem,
the space $\mathfrak{H}^{1}$  has the form
\begin{equation}
\label{H^1}
\mathfrak{H}^{1}=\{f \in \mathfrak{H}^{0}\colon \int_0^Te^{-i(b-c)t}f(t)dt=0\}.
\end{equation}
This space as subspace of $\mathfrak{H}^{0}$ has codimension 2 if $b\neq c$, and codimension 1 if $b=c$.
The  second-order forms  $A^2_{34}$
and $A^2_{24}$ can be rewritten as
\[
A^2_{34}\langle f\rangle=\frac 12 v^\ast_{23} v_{24}\left(\int_0^T\int_0^T K(t_1,t_2)f(t_1)f(t_2)dt_1dt_2\right)
\]
and
\[A^2_{24}\langle f\rangle=\frac 12 v_{23} v_{34}\left(\int_0^T\int_0^TK(t_1,t_2)f(t_1)f(t_2)dt_1dt_2\right),
\]
where 
$$
K(t_1,t_2)=e^{-i(b-c)\min{(t_1,t_2)}}.
$$
Define the set of controls $\mathfrak{H}^{2}=\{f\in \mathfrak{H}^{1}\colon J_O^{(4)}(f_0)(f,\ldots,f)=0\}$. It has the form
$$
\mathfrak{H}^{2}=\{f \in \mathfrak{H}^{1}\colon  
\int_0^T\int_0^T K(t_1,t_2) f(t_1)f(t_2)dt_1dt_2=0\}.
$$
So for any $\delta f\in \mathfrak{H}^{1}\setminus \mathfrak{H}^{2}$ it holds that $J^{(4)}(f_0)(\delta f,\ldots,\delta f)<0$.
Formulas~(\ref{J(2n-1)}) and~(\ref{J(2n)}) for $n=3$ also imply that for any $\delta f\in \mathfrak{H}^{2}$ it holds that $J^{(5)}(f_0)(\delta f,\ldots,\delta f)=0$ 
and
$$
J^{(6)}(f_0)(\delta f,\ldots,\delta f)=(\lambda_1-\lambda_4)|A^3_{14}\langle \delta f\rangle|^2+(\lambda_3-\lambda_4)|A^3_{34}\langle \delta f\rangle|^2+(\lambda_2-\lambda_4)|A^3_{24}\langle \delta f\rangle|^2,
$$
where $(\lambda_1-\lambda_4)>0$. Hence $f_0\equiv 0$ is a trap at least of the fifth order.
This completes the proof.

Now consider in more detail the case  $b=c$, i.e. the case when the free Hamiltonian has the matrix form $H_0=\mathrm{diag}(a,b,b,b)$
and the interaction Hamiltonian $V$ has the matrix form~(\ref{V}) such that relation~(\ref{condvv}) is satisfied. For $v_{13}=v_{24}=0$ such systems were considered in~\cite{VolkovPechenUMN}. In  this case, the sets of controls $\mathfrak{H}^{1}$ and $\mathfrak{H}^{2}$ coincide,
\begin{equation}
\label{H=H}
\mathfrak{H}^{1}=\mathfrak{H}^{2}=\{f \in \mathfrak{H}^{0}\colon \int_0^Tf(t)dt=0\}.
\end{equation}
If $\delta f \in \mathfrak{H}^{0}\setminus \mathfrak{H}^{1}$, then it holds that
$
J^{(2)}(f_0)(\delta f,\delta f)<0
$
and for all $\delta f\in\mathfrak{H}^{1}$ it holds that
$
J^{(k)}(f_0)(\delta f,\ldots,\delta f)=0$ for any $k\in \{2,3,4,5\}$. 
Note that
\begin{eqnarray}
\label{A324}
A^3_{24}\langle f\rangle=v_{12}^\ast(v_{13}v_{34} +v_{12}v_{24})\int_0^1dt_1\int_0^{t_1}dt_2\int_0^{t_2}dt_3e^{i(a-b)(t_1-t_2)}f(t_1)f(t_2)f(t_3)\\+\frac1{6}v_{24}(|v_{24}|^2+|v_{34}|^2+|v_{23}|^2)\left(\int_0^Tf(t)dt\right)^3.
\end{eqnarray}
Because of~(\ref{condvv}) the first term in the right side of~(\ref{A324}) vanishes. Hence, for any $\delta f \in \mathfrak{H}^{1}$ the equality $A^3_{24}\langle \delta f\rangle=0$ and 
similarly  the equality $A^3_{34}\langle \delta f\rangle=0$ are fulfilled, which means that
\[
\frac 1{6!}J^{(6)}(f_0)(\delta f,\ldots,\delta f)=(\lambda_1-\lambda_4)|A^{3}_{14}\langle \delta f\rangle|^2\geq 0,
\]
where
\[
A_{14}^3\langle f \rangle=\frac 16(v_{12}v_{23}v_{34}+
v_{13}v^\ast_{23}v_{24})
\int_0^T\int_0^T\int_0^TR(t_1,t_2,t_3)f(t_1)f(t_2)f(t_3)dt_1dt_2dt_3,
\]
with
\[
R(t_1,t_2,t_3)=e^{-i(a-b)\max(t_1,t_2,t_3)}.
\]
Therefore, $f_0\equiv 0$ is exactly a trap of fifth order. The fine structure of the control subspace~(\ref{H=H}) might be essential for the efficiency of optimization over the quantum control landscape.

\section{GRAPE based analysis of the unconstrained quantum control landscapes}\label{Sec:3}

In this section, we perform a numerical analysis of the control landscapes for some examples of the four-level quantum systems having the null control as a five-order trap. For this analysis, we extend the method of investigation of quantum control landscapes~\cite{PechenTannor2012} which is based on a fixed step GRAPE optimization~\cite{Khaneja_Reiss_Kehlet_Schulte-Herbruggen_Glaser_2005,GRAPE}. While quasi-Newton methods can be much more efficient~\cite{Eitan_Mundt_Tannor_2011}, we use a simple fixed step gradient search to reveal features of the quantum control landscape. 

The first system which we consider is 
\begin{equation}
\label{ExSyst1}
S_1:\quad H_0=\left(\begin{array}{cccc}
-2& 0& 0 & 0 \\
0&  1& 0 & 0 \\
0&  0& 1 & 0 \\
0&  0& 0 & 0 \\
\end{array} \right),\quad 
V=\left(\begin{array}{cccc}
0&  -2 & 1  & 0 \\
-2 & 0 & 1 & 1 \\
1 & 1 & 0 &  2 \\
0 & 1 & 2 & 0 \\
\end{array} \right).
\end{equation}
Its energy level structure and allowed transitions are shown on the left part of Fig.~\ref{Fig}.In~\cite{MyachkovaPechen}, controllability of such  system (up to a renumeration of the energy levels) was investigated and it was found that such system is controllable. The structure of the control space as a chain of embedded sets $\mathfrak{H}^{0}\supset \mathfrak{H}^{1}\supset \ldots$, which determines the order of the trap,  
has a sufficiently complex form for systems like $S_1$. The subspace $\mathfrak{H}^{1}$, on which the second derivative vanishes due to~(\ref{H^1}),  has codimension 2.

The second system is a controlled system with a chain interaction Hamiltonian. Control landscapes for such systems were theoretically investigated in~\cite{VolkovPechenUMN}. The complete controllability of such quantum systems was proved in~\cite{SchirmerFuSolomon}. We consider the  anharmonic system with the following explicit chain Hamiltonian
\begin{equation}
\label{ExSyst2}
S_2:\quad H_0=\left(\begin{array}{cccc}
0& 0 & 0 & 0 \\
0&  1& 0  & 0 \\
0&  0& 1 &  0 \\
0&  0& 0 & 1 \\
\end{array} \right),\quad 
V=\left(\begin{array}{cccc}
0&  -2 & 0  & 0 \\
-2 & 0 & 1 & 0 \\
0 & 1 & 0 &  2 \\
0 & 0 & 2 & 0 \\
\end{array} \right).
\end{equation}
Its energy level structure and allowed transitions are shown on the right part of Fig.~\ref{Fig}. The structure of the control space as a chain of embedded sets $\mathfrak{H}^{0}\supset \mathfrak{H}^{1}\supset \ldots$, which determines the order of the trap,  has a simple form~(\ref{H=H}) for  systems like $S_2$. The subspace $\mathfrak{H}^{1}$ on which the second derivative vanishes  has a codimension 1.

For these systems, we consider the target observable $O=\mathrm{diag}(\lambda_1,\lambda_2,\lambda_3,0)$, where $\lambda_1>0$, $\lambda_2<0$, $\lambda_3<0$ (we explicitly consider $O=\mathrm{diag}(1,-6,-6,0)$) and the initial state $\rho_0=|4\rangle \langle 4|$. For the system $S_2$, which exhibits a novel behavior, we also consider another target observable $O^*=\mathrm{diag}(1,6,-6,0)$. It follows from the results of~\cite{VolkovPechenUMN} that for the target observable $O^*$ and the quantum system $S_2$ the null control $f_0$ is also a fifth order trap for   the objective functional $J_{O^*}$.

The algorithm for the analysis of the control landscapes of the corresponding four-level quantum systems is based on the fixed step GRAPE approach and has the following form. We divide the entire control interval $[0,T]$ into $M$ equal subintervals, each of duration $\Delta t = T/M$, and consider the control as constant on each subinterval. The resulting control is a piecewise constant function of the form 
\[
f_C(t)=\sum_{k=1}^M c_k\chi_{[t_k,t_{k+1}]}(t),
\]
where $C = (c_1,\dots,c_M)\in\mathbb R^M$ is an $M$-dimensional vector, $\chi_{[t_k,t_{k+1}]}(t)$ is the characteristic function of the interval $[t_k,t_{k+1}]$, and $t_k = \Delta t(k-1)$. Vector $C$ is a new, already finite-dimensional, control. The control objective of maximizing the mean value of the observable $O$ then becomes maximizing the function $\mathcal{J}_O\colon \mathbb{R}^M\to \mathbb{R}$ defined as
$\mathcal{J}_O(C) = J_O(f_C)={\rm Tr}[U_T^{f_C} \rho_0(U_{T}^{f_C})^{\dagger}O]$, where $U^{f_C}_T=U_M\dots U_2U_1$ and $U_k = e^{-i(H_0+c_kV)\Delta t}$. Gradient of $\mathcal{J}_O(C)$ is defined by partial derivatives,
\begin{eqnarray*}
{\rm grad}\mathcal{J}_O(C)= \left(\frac{\partial \mathcal{J}_O(C)}{\partial c_1},\ldots,\frac{\partial \mathcal{J}_O(C)}{\partial c_M}\right).
\end{eqnarray*}
We use for the partial derivative with respect to $c_k$ the following linear in $\Delta t$ approximation,
\[
\label{gradient}
\frac{\partial \mathcal{J}_O(C)}{\partial c_k} \approx 2\Delta t \times \rm{Im}\Bigl[{\rm Tr}\Bigl(W_k^{\dagger}VW_k\rho_0W_M^{\dagger}OW_M\Bigr)\Bigr]=({\rm grad}_{\rm lin}\mathcal{J}_O(C))_k,
\]
where $W_k = U_kU_{k-1}\dots U_2U_1$.

Our control landscape analysis algorithm starts with constructing a random initial control $C_1$, whose components are generated with a uniform distribution in the interval $[-l, l]$ with some $l>0$. On $i$-th iteration, we compute the linear approximation ${\rm grad}_{\rm lin}\mathcal{J}_O(C_i)$ of the gradient and update the control as $C_{i+1}=C_{i}+\varepsilon\cdot{\rm grad}_{\rm lin}\mathcal{J}_O(C_i)$, where $\varepsilon>0$ is a fixed small number. The loop continues either until $\mathcal{J}_O$ reaches the fidelity value $J_{\mathrm{stop}}=1-I_{\mathrm{err}}$, where $I_{\mathrm{err}}$ is some threshold describing admissible deviation from the global maximum objective, or until a predefined maximum admissible number of iterations $K_{\mathrm{stop}}$ is reached. 

For simulations we use final time $T=20$, number of control vector components $M=100$, initial control vectors generated randomly with uniform distribution in the hypercube $[-l,l]^M$, maximum admissible number of iterations $K_{\mathrm{stop}} = 2000$, step size $\varepsilon = 0.02$, and admissible deviation from the global maximum objective $I_{\mathrm{err}} = 10^{-3}$. For this $I_{\mathrm{err}}$ we have $J_{\mathrm{stop}}=0.999$. If the algorithm starts from some initial control and can not reach the value $J_{\mathrm{stop}}$ with less than $K_{\mathrm{stop}}$ iterations, we call this run as failed. To get a reasonable statistics, for each $l$ we generate $L=500$ random initial controls in the hypercube $[-l,l]^M$ and compute the fraction of failed runs among the corresponding $500$ GRAPE optimization runs.

The results for the system $S_1$ with Hamiltonian~(\ref{ExSyst1}) are shown on the right subplot of Fig.~\ref{Fig:S1S2}  and in the upper row on Fig.~\ref{Fig:S1S2Grad}. They show that in the small vicinity of the null control higher-order trap, most runs are failed, ($98.8\% $ of runs are failed for $l=0.1$ and $88\%$ for $l=0.2$). With increase of $l$ the number of failed runs rapidly decreases to zero for $l=1.4$. This behavior is similar to the case of $\Lambda$-atom~\cite{Volkov_Myachkova_Pechen_2025}. On Fig.~\ref{Fig:S1S2Grad} (upper row, left subplot) we show for the system $S_1$ for $l=1$, so that when almost all runs are successful, an example of the behavior of norm of the approximate gradient (i.e., $\|{\rm grad}_{\rm lin}\mathcal{J}_O(C_i)\|$) of the objective and objective value vs iteration number of the algorithm. Objective value monotonically increases towards $J_{\mathrm{stop}}=0.999$. The obtained optimal control is quite complicated (right subplot).

The results for the system $S_2$ with the chain Hamiltonian~(\ref{ExSyst2}), which are done for target observables $O$ and $O^*$, are shown on Fig.~\ref{Fig:S1S2} and in the bottom row on Fig.~\ref{Fig:S1S2Grad}. They are different from the case of  the system $S_1$. We see that similarly to the case with the system $S_1$, in a small vicinity of the higher-order trap the fraction of the failed runs for both target observables $O$ and $O^*$ is close to $100\%$. The fraction of failed runs decreases to its minimum around $l\approx 2$. However, even at the minimum value of the fraction of failed runs is quite large, it is about $20\%$. This behavior is in contrast to the behavior found for three-level $\Lambda$-type systems in~\cite{Volkov_Myachkova_Pechen_2025}. As another interesting phenomenon, we observe after the minimum around $l\approx 2$ a slight increase of the fraction of failed runs for large values of $l$. For $l=2$, when most runs are successful, on Fig.~\ref{Fig:S1S2Grad} (bottom row, left subplot) we show an example of the behavior of the norm of the approximate gradient of the objective function and of the objective value vs iteration number of the algorithm. Objective value oscillates in the beginning and then increases towards $J_{\mathrm{stop}}=0.999$. The obtained optimal control is quite complicated (right subplot).

The main conclusion from this analysis is that optimization behavior is essentially different for the two considered kinds of four-level quantum systems. For a $\VV$ system we observe sharp increase of the optimization efficiency up to $100\%$  at certain distance with $l\approx 1$ from the null control (less value of the fraction of failed runs corresponds to a higher value of the optimization efficiency; zero fraction of failed runs corresponds to $100\%$ optimization efficiency). And moreover, the optimization efficiency increases from $10\%$ to $100\%$ in a short range of $l=0.2$ to $l=1$. For the system with the chain interaction the observed increase in efficiency is less significant: it goes up to about $80\%$ and then even begins to slowly decrease, and much slower in the initial increase: it goes from $10\%$ to $80\%$ for from $l=0.2$ to $l=1.5$. 

A possible origin of this effect might be related to the structure of the subspace where second derivative of the objective functional vanishes. For a strongly degenerate system $S_2$ with a chain Hamiltonian, the control subspace $\mathfrak{H}^{1}$ on which the second derivative of the objective functional is zero has codimension 1. For the $\VV$ system, such subspace has codimension 2. In the paper~\cite{Volkov_Myachkova_Pechen_2025}, the behavior of the search algorithm in the case of three-level systems was similar to the case of  the $\VV$ system, while the algorithm running time was significantly affected by the order of the trap.
For all cases considered in~\cite{Volkov_Myachkova_Pechen_2025}, the subspace $\mathfrak{H}^{1}$ has  codimension 2. Thus, a hypothesis naturally arises that optimization is influenced not only by the order of the trap, but also by the structure of the control subspace $\mathfrak{H}^{1}$ which determines the order of the trap. Hence we expect that the significant differences between the cases of systems $S_1$ and $S_2$ can be related to the different codimensions of the subspace $\mathfrak{H}^{1}$ for these systems.

\begin{figure}
\includegraphics[width=\linewidth]{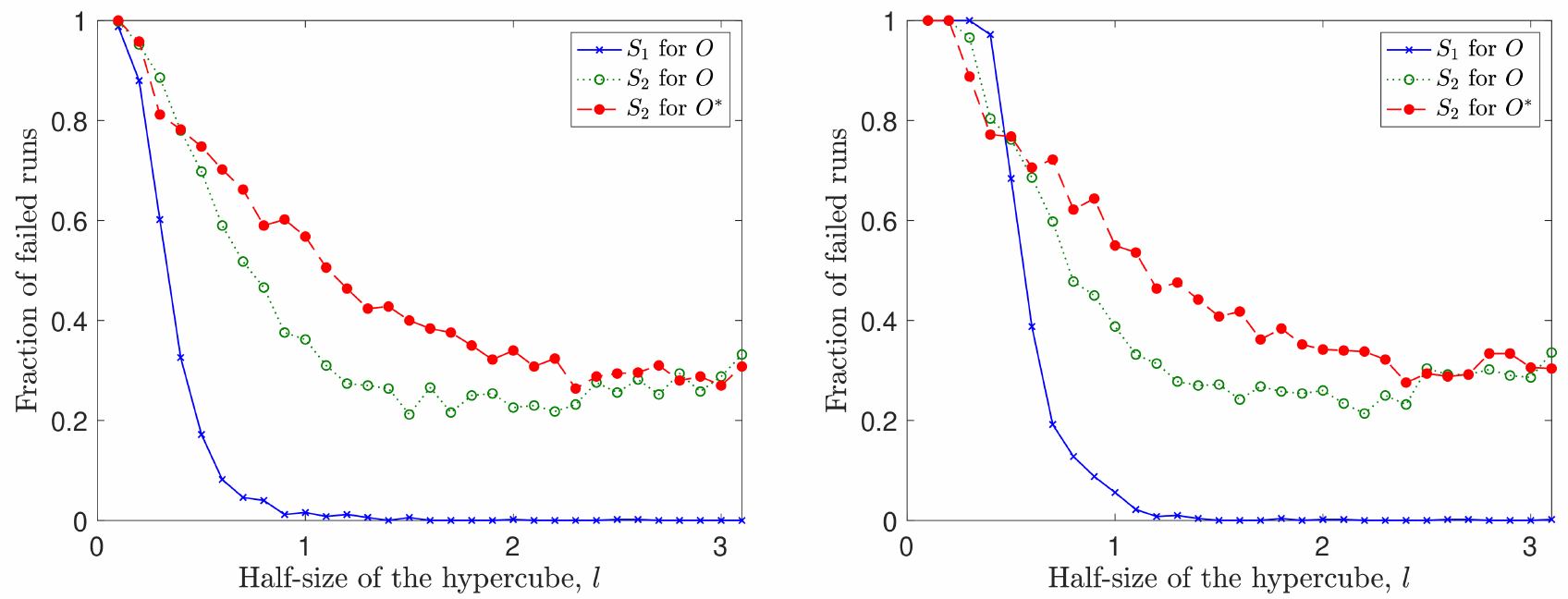}
\caption{Fraction of failed runs for $l=0.1,0.2,0.3,\dots,3.1$ for the system $S_1$ (blue solid line) and $S_2$ (green dotted and red dashed lines for the two target observables $O$ and $O^*$, respectively). Left: Unconstrained optimization using GRAPE. Right: Constrained optimization using GPM. The target observable is $O=\mathrm{diag}(1,-6,-6,0)$ for $S_1$ and $S_2$, and also $O^*=\mathrm{diag}(1,6,-6,0)$ for $S_2$. The initial state is $\rho_0=|4\rangle \langle 4|$. The parameters are $T=20$, $M=100$,  $K_{\mathrm{stop}} = 2000$, $\varepsilon = 0.02$, and $I_{\mathrm{err}} = 10^{-3}$. For each point $L=500$ runs are performed. The behaviour of the fraction of failed runs for both systems in essentially different.}
\label{Fig:S1S2}
\end{figure}

\begin{figure}
\includegraphics[width=\linewidth]{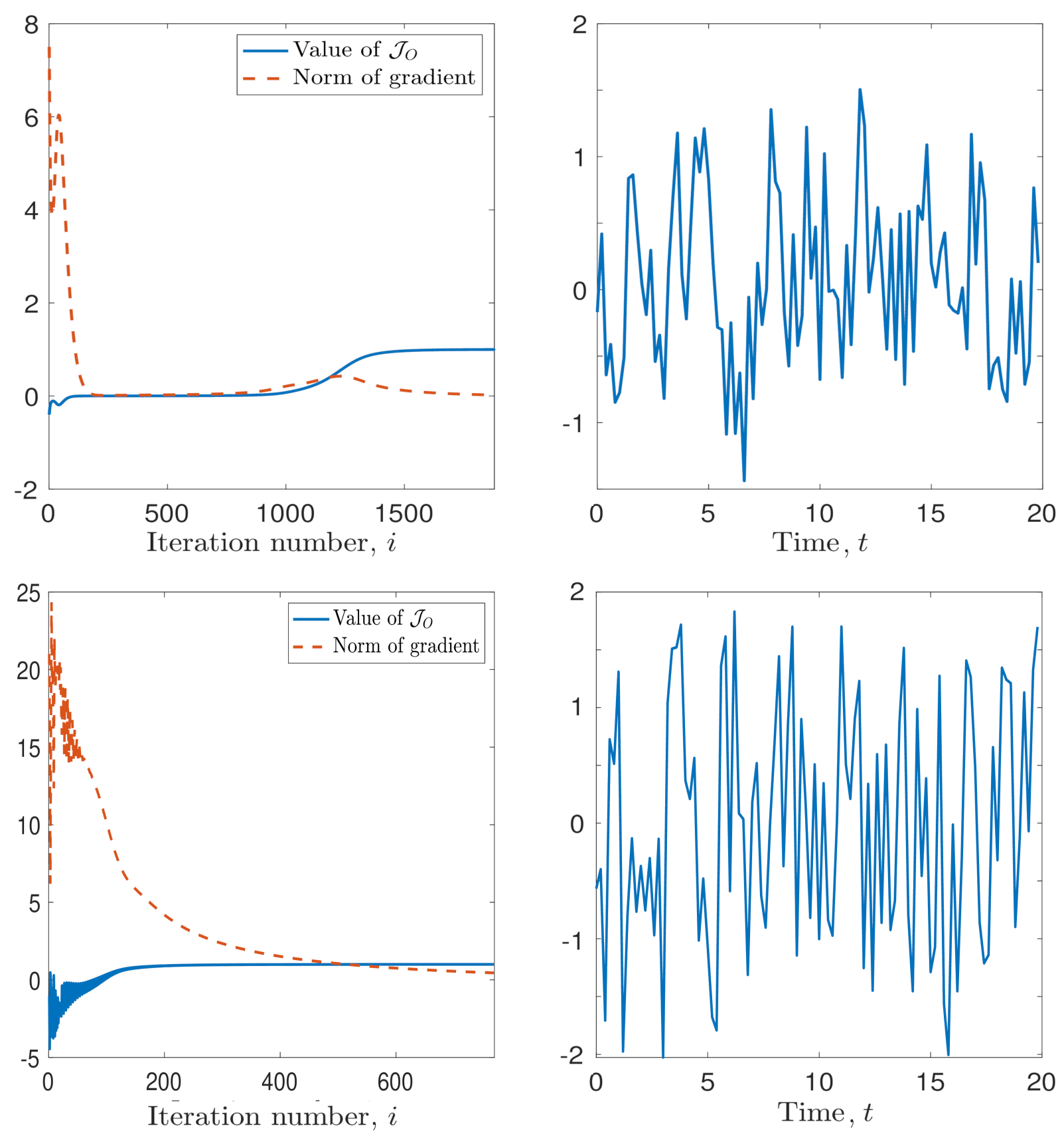}
\caption{An example of the optimization results with target observable $O$. Upper row for the system $S_1$ and $l=1$. Bottom row for the system $S_2$ and $l=2$.  Other parameters are the same as for Fig.~\ref{Fig:S1S2}. Left: behavior of the objective value and norm of the gradient vs iteration number. Totally 1880 iterations were required for $S_1$ and 770 for $S_2$. Right: The obtained optimized controls vs time.}
\label{Fig:S1S2Grad}
\end{figure}

\section{GPM based analysis of the constrained quantum control landscapes}
\label{Sec:GPM}

In the analysis above, the initial randomly generated controls were constrained to the hypercube $[-l,l]^M$, but during GRAPE optimization the controls were not restricted to this hypercube and were allowed to escape it. An interesting question is about constrained control landscape properties, when the controls should be restricted to the hypercube during {\it all} iterations of the GRAPE run. For this case, one can use gradient projection method (GPM) known in optimization theory~\cite{Goldstein1964,LevitinPolyak1966} and developed for constrained optimization of quantum systems in~\cite{Oza2005} (summarized in Section~5 and Appendix~B). In~\cite{Oza2005}, one considers the most general case when controls are represented by Kraus maps and experimentalist is allowed to apply different Kraus maps to the system. It can be done by using coherent and incoherent controls, quantum measurements, or other actions. Then, parameterizing Kraus maps by points of a complex Stiefel manifold, the constrained optimization of quantum systems is considered as optimization over complex Stiefel manifolds with additional constraints, where Kraus maps are used as explicit controls to manipulate quantum systems. In~\cite{Oza2005}, constrained gradient projection optimization for open quantum systems was considered in full generality, including theoretical derivation of the projection and numerical simulations, when coherent and incoherent controls are represented in the kinematic description as Kraus maps. Recently, GPM was used for control of some physical models of quantum systems~\cite{MorzhinPechenJPA2025}.

The GPM based method for the analysis of quantum control landscapes was developed in~\cite{PechenLJM2025}. In this case, in difference from the unconstrained GRAPE based analysis, for the constraint  $|C(k)|\le l$ that corresponds to $Q=\{C\in \mathbb R^M\,:\, |C(k)|\le l \textrm{ for } k=1,\dots,M\}$ being a hypercube in $\mathbb R^M$, the control on $i+1$ iteration is obtained using the projection ${\rm Pr}_Q$ on $Q$ as
\[
C_{i+1}={\rm Pr}_Q(C_i+\epsilon_i{\rm grad}{\mathcal J_O(C_i)}),
\]
Explicitly is has the form
\begin{equation}\label{Eq:Projection1}
    C_{i+1}(k)=\left\{
    \begin{array}{ll}
    -A, &  \textrm{ if } C_{i}(k)+\epsilon_i \frac{\partial \mathcal J_O(C_i)}{\partial C_i(k)}\le -l,\\
    C_i(k)+\epsilon_i\frac{\partial \mathcal J_O(C_i)}{\partial C_i(k)},     & \textrm{ if } -l< C_{i}(k)+\epsilon_i\frac{\partial \mathcal J_O(C_i)}{\partial C_i(k)}< l,\\
    A, &  \textrm{ if } C_{i}(k)+\epsilon_i \frac{\partial \mathcal J_O(C_i)}{\partial C_i(k)}\ge l.
    \end{array}\right.
\end{equation}
Here for the gradient we can use either the exact or the linear approximation. In the numerical analysis below, similarly to the unconstrained case we choose to use the linear form~(\ref{gradient}). 

The results are shown on the right subplot of Fig.~\ref{Fig:S1S2}. Qualitatively the behavior of the fractions of failed runs is similar to the unconstrained case with GRAPE optimization, but quantitatively it differs, as the fraction of failed runs for GPM constrained optimization is typically higher than that for GRAPE optimization. However, the general conclusion about different behavior of optimization efficiency for these two kinds of systems is the same.

\section{Conclusions}\label{Sec:4}
We study both theoretically and numerically quantum control landscapes for maximizing mean value of a Hermitian observable for special four-level quantum systems for which null control is proved to be a five-order trap. After theoretically establishing the trapping behavior, we use a numerical method based on the GRAPE algorithm to investigate the level of efficiency of optimizing the control objective over the quantum control landscape depending on the average distance of the initial controls from the the null control. Starting from random piecewise constant controls uniformly generated in some hypercube around the higher-order trap, we compute the fraction of GRAPE runs which are not successful, i.e. which do not allow to obtain a predefined high enough value of the fidelity within a maximum allowed number of iterations. We study the dependence of this fraction on the size of the hypercube which determines how far on average the initial controls are from the higher-order trap. Both cases of unconstrained optimization using GRAPE and constrained optimization using GPM are considered. While for the system $S_1$ the behavior of the fraction of  failed runs in similar to the $\Lambda$-type system~\cite{Volkov_Myachkova_Pechen_2025}, the discovered in this work behavior of the fraction of failed runs, and hence of the optimization efficiency, for a $S_2$ ladder system which was theoretically studied in~\cite{VolkovPechenUMN} is essentially different --- the increase in the optimization efficiency is much less and slower when moving away from the null control five-order trap, and even the optimization efficiency starts to decrease at some distance from the null control, --- that requires a further detailed investigation. This sharp difference might be related to the structure of the control space as a chain of embedded sets $\mathfrak{H}^{0}\supset \mathfrak{H}^{1}\supset \ldots$, which determines the order of the trap order, and in particular by the subspace where $\mathfrak{H}^{1}$ where second derivative of the objective functional is zero.

\funding{This work was supported by the RSF grant no. 22-11-00330-P, \url{https://rscf.ru/project/22-11-00330/}}

\section*{Data availability statement}
All data that support the findings of this study are included within the article (and any supplementary files).

\end{document}